\documentclass[11pt]{llncs}

\usepackage{amssymb}
\usepackage{amsmath}
\usepackage{amsthm}
\usepackage{textcomp}
\usepackage{array}
\usepackage[latin5]{inputenc}
\usepackage{lineno}
\usepackage{verbatim} 
\usepackage{setspace}
\usepackage{endnotes}
\usepackage{cite}

\newcommand{\bra}[1]{\langle #1|}
\newcommand{\ket}[1]{|#1\rangle}

\newcommand{\cent}[0]{\mbox{\textcent}}
\newcommand{\dollar}[0]{\$}

\title{Superiority of exact quantum automata for promise problems}

\author{Andris Ambainis\thanks{Ambainis was supported by ESF project 1DP/1.1.1.2.0/09/APIA/VIAA/044,
FP7 Marie Curie International Reintegration Grant PIRG02-GA-2007-224886 and FP7 FET-Open project QCS.} \and 
Abuzer Yakary{\i}lmaz\thanks{Yakary{\i}lmaz was partially supported by the Scientific and Technological
Research Council of Turkey (T\"{U}B\.ITAK) with grant 108E142 and FP7 FET-Open project QCS.}}
\institute{University of Latvia, Faculty of Computing, Raina bulv. 19, Riga, LV-1586, Latvia 
\\ ~ \\
\texttt{\{ambainis,abuzer\}@lu.lv}
\\ ~ \\
\today
}

\begin{document}

\maketitle

\pagenumbering{arabic}

\thispagestyle{empty}

% AAAAAAAAAAAAAAAAAAAAAAAAAAAAAAAAAAAAAAAAAAAAAAAAAAAAAAAAAAAAAAAAAAAAAAAAAAAAAAAA % 
% AAAAAAAAAAAAAAAAAAAAAAAAAAAAAAAAAAAAAAAAAAAAAAAAAAAAAAAAAAAAAAAAAAAAAAAAAAAAAAAA % 
\begin{abstract}
In this note, we present an infinite family of promise problems which can be solved exactly
by just tuning transition amplitudes of a two-state quantum finite automata operating in realtime mode,
whereas the size of the corresponding classical automata grow without bound.
\\ \\
{\textbf{Keywords:} exact quantum computation, promise problems, succinctness,
quantum finite automaton, classical finite automaton}
\end{abstract}
% AAAAAAAAAAAAAAAAAAAAAAAAAAAAAAAAAAAAAAAAAAAAAAAAAAAAAAAAAAAAAAAAAAAAAAAAAAAAAAAA % 
% AAAAAAAAAAAAAAAAAAAAAAAAAAAAAAAAAAAAAAAAAAAAAAAAAAAAAAAAAAAAAAAAAAAAAAAAAAAAAAAA % 

% SSSSSSSSSSSSSSSSSSSSSSSSSSSSSSSSSSSSSSSSSSSSSSSSSSSSSSSSSSSSSSSSSSSSSSSSSSSSSSSS %
% SSSSSSSSSSSSSSSSSSSSSSSSSSSSSSSSSSSSSSSSSSSSSSSSSSSSSSSSSSSSSSSSSSSSSSSSSSSSSSSS %
% SSSSSSSSSSSSSSSSSSSSSSSSSSSSSSSSSSSSSSSSSSSSSSSSSSSSSSSSSSSSSSSSSSSSSSSSSSSSSSSS %
\section{Introduction}
% SSSSSSSSSSSSSSSSSSSSSSSSSSSSSSSSSSSSSSSSSSSSSSSSSSSSSSSSSSSSSSSSSSSSSSSSSSSSSSSS %
% SSSSSSSSSSSSSSSSSSSSSSSSSSSSSSSSSSSSSSSSSSSSSSSSSSSSSSSSSSSSSSSSSSSSSSSSSSSSSSSS %
% SSSSSSSSSSSSSSSSSSSSSSSSSSSSSSSSSSSSSSSSSSSSSSSSSSSSSSSSSSSSSSSSSSSSSSSSSSSSSSSS %

The exact quantum computation has been widely examined for partial (promise) and total functions 
(e.g. \cite{BH97,BV97,BHCMW98,BCWZ99,Kl00,BV03,MNYW05,FI09,YFSA10}).
On the other hand, in automata theory, only two results have been obtained:

\begin{enumerate}
\item[(i)] Klauck \cite{Kl00} has shown that realtime quantum finite automata (QFAs) cannot be 
more concise than realtime deterministic finite automata (DFAs) \footnote{The proof was basically given for Kondacs-Watrous realtime QFA model 
\cite{KW97} but it can be extended for any model of realtime QFAs including the most general ones 
\cite{Ci01,BMP03,Hi10,YS11A}.}
in case of language recognition,
\item[(ii)] Murakami et. al. \cite{MNYW05} have shown that there is a promise problem solvable by
quantum pushdown automata but not by any deterministic pushdown automata.
\end{enumerate}

In this note, we consider succinctness of realtime QFAs for promise problems.
We present an infinite family of promise problems which can be solved exactly
by just tuning transition amplitudes of a two-state rtQFAs,
whereas the size of the corresponding classical automata grow without bound.

% SSSSSSSSSSSSSSSSSSSSSSSSSSSSSSSSSSSSSSSSSSSSSSSSSSSSSSSSSSSSSSSSSSSSSSSSSSSSSSSS %
% SSSSSSSSSSSSSSSSSSSSSSSSSSSSSSSSSSSSSSSSSSSSSSSSSSSSSSSSSSSSSSSSSSSSSSSSSSSSSSSS %
% SSSSSSSSSSSSSSSSSSSSSSSSSSSSSSSSSSSSSSSSSSSSSSSSSSSSSSSSSSSSSSSSSSSSSSSSSSSSSSSS %
\section{Background}
% SSSSSSSSSSSSSSSSSSSSSSSSSSSSSSSSSSSSSSSSSSSSSSSSSSSSSSSSSSSSSSSSSSSSSSSSSSSSSSSS %
% SSSSSSSSSSSSSSSSSSSSSSSSSSSSSSSSSSSSSSSSSSSSSSSSSSSSSSSSSSSSSSSSSSSSSSSSSSSSSSSS %
% SSSSSSSSSSSSSSSSSSSSSSSSSSSSSSSSSSSSSSSSSSSSSSSSSSSSSSSSSSSSSSSSSSSSSSSSSSSSSSSS %

Throughout the paper,
\begin{enumerate}
\item[(i)]
$ \Sigma $ denotes the input alphabet not containing left- and right-end markers 
($ \cent $ and $ \dollar $, respectively) and $ \tilde{\Sigma} = \Sigma \cup \{ \cent,\dollar \} $,
\item[(ii)] $ \varepsilon $ is the empty string,
\item[(iii)] $ w_i $ is the $ i^{th} $ symbol of a given string $ w $, and
\item[(iv)] $ \tilde{w} $ represents the string $ \cent w \dollar $, for $ w \in \Sigma^{*} $.
\end{enumerate}
Moreover, all machines presented in the paper operate in realtime mode.
That is, the input head moves one square to the right in each step and the computation stops after reading $ \dollar $.

A promise problem is a pair $ A = (A_{yes},A_{no}) $, 
where $ A_{yes},A_{no} \subseteq \Sigma^{*} $ and $ A_{yes} \cap A_{no} = \emptyset $ \cite{Wa09}.
A promise problem $ A = (A_{yes},A_{no}) $ is solved exactly by a machine $ \mathcal{M} $
if each string in $  A_{yes} $ (resp., $  A_{no} $) is accepted (resp., rejected) exactly by $ \mathcal{M} $.
Note that, if $ \overline{ A_{yes} } =   A_{no} $, 
this is the same as the recognition of a language ($ A_{yes} $).

We give our quantum result with the most restricted of the known QFA model, i.e.
\textit{Moore-Crutchfield quantum finite automaton} (MCQFA) \cite{MC00},
(see \cite{YS11A} for the definition of the most general QFA model).

A MCQFA is a 5-tuple 
\begin{equation*}
	\mathcal{M} = ( Q , \Sigma, \{ U_{\sigma} \mid \sigma \in \tilde{\Sigma} \},q_{1},Q_{a} ),
\end{equation*}
where $ Q $ is the set of states, $ q_{1} $ is the initial state, $ Q_{a} \subseteq Q $ and 
is the set of accepting states,
and $ U_{\sigma} $'s are unitary operators.
The computation of a MCQFA on a given input string $ w \in \Sigma^{*} $ can be traced by a $ |Q| $-dimensional vector.
This vector is initially set to $ \ket{v_{0}} = (1~0~\cdots~0)^{T} $ and evolves according to
\[ 
	\ket{v_{i}} = U_{\tilde{w}_{i}} \ket{v_{i-1}}, ~~~~ 1 \leq i \leq |\tilde{w}|.
\]
At the end of the computation, $ w $ is accepted (resp., rejected) with probability 
$  || P_{a}v_{|\tilde{w}|} ||^{2}  $
(resp., $  || P_{r}v_{|\tilde{w}|} ||^{2}  $),
where $ P_a = \sum_{q \in Q_{a}} \ket{q}\bra{q} $ and $ P_r = I-P_a $.
If we replace the unitary operation with a zero-one left stochastic operator,
we obtains a realtime DFA (which we call simply a DFA).

% SSSSSSSSSSSSSSSSSSSSSSSSSSSSSSSSSSSSSSSSSSSSSSSSSSSSSSSSSSSSSSSSSSSSSSSSSSSSSSSS %
% SSSSSSSSSSSSSSSSSSSSSSSSSSSSSSSSSSSSSSSSSSSSSSSSSSSSSSSSSSSSSSSSSSSSSSSSSSSSSSSS %
% SSSSSSSSSSSSSSSSSSSSSSSSSSSSSSSSSSSSSSSSSSSSSSSSSSSSSSSSSSSSSSSSSSSSSSSSSSSSSSSS %
\section{The main results}
% SSSSSSSSSSSSSSSSSSSSSSSSSSSSSSSSSSSSSSSSSSSSSSSSSSSSSSSSSSSSSSSSSSSSSSSSSSSSSSSS %
% SSSSSSSSSSSSSSSSSSSSSSSSSSSSSSSSSSSSSSSSSSSSSSSSSSSSSSSSSSSSSSSSSSSSSSSSSSSSSSSS %
% SSSSSSSSSSSSSSSSSSSSSSSSSSSSSSSSSSSSSSSSSSSSSSSSSSSSSSSSSSSSSSSSSSSSSSSSSSSSSSSS %

Let $ A_{yes}^{k} = \{ a^{i2^{k}} \mid i \mbox{ is a nonnegative \textit{even} integer} \} $ and
$ A_{no}^{k} = \{ a^{i2^{k}} \mid i \mbox{ is a positive \textit{odd} integer} \} $ be two unary languages, 
where $ k $ is a positive integer.
We will show that a two-state MCQFA can solve promise problem $ A^{k}=( A_{yes}^{k},A_{no}^{k} ) $, but
any DFA (and so any PFA) must have at least $ 2N $ states to solve the same problem exactly.

\begin{theorem}
	Promise problem $ A^{k}=( A_{yes}^{k},A_{no}^{k} ) $ can be solved by a two-state MCQFA $ \mathcal{M}_{k} $ exactly. 
\end{theorem}
\begin{proof}
	We will use a well-known technique given in \cite{AF98}.
	Let $ N = 2^{k} $ and 
	$ \mathcal{M}_{k} = ( Q , \Sigma, \{ U_{\sigma} \mid \sigma \in \tilde{\Sigma} \},q_{1},Q_{a}) $,
	where $ Q = \{q_{1},q_{2}\} $, $ \Sigma = \{a\} $, $ Q_{a} = \{q_{1}\} $, 
	$ U_{\cent} = U_{\dollar} = I $, and $ U_{a} $ is a rotation in $ \ket{q_{1}} $-$ \ket{q_{2}} $ plane
	with angle $ \theta = \frac{\pi}{2N} $, i.e.,
	\begin{equation*}
		U_{a} = \left(  \begin{array}{lr}
			\cos \theta & ~-\sin \theta \\
			\sin \theta & \cos \theta
		\end{array}  \right).
	\end{equation*}
	The computation begins with $ \ket{q_{1}} $ and after reading each block of $ N $ $ a $'s,
	the following pattern is followed by $ \mathcal{M}_{k} $:
	\begin{equation*}
		\ket{q_{1}} \overset{a^{N}}{\longrightarrow} 
		\ket{q_{2}} \overset{a^{N}}{\longrightarrow} 
		-\ket{q_{1}} \overset{a^{N}}{\longrightarrow} -\ket{q_{2}}
		\overset{a^{N}}{\longrightarrow} \ket{q_{1}} 
		\overset{a^{N}}{\longrightarrow} \cdots .
	\end{equation*}
	Therefore, it is obvious that $ \mathcal{M}_{k} $ solves promise problem $ A^{k} $ exactly.
\end{proof}
\begin{lemma}
	Any DFA solving $ A^{k}=( A_{yes}^{k},A_{no}^{k} ) $ exactly must have at least $ 2^{k+1} $ states.
\end{lemma}
\begin{proof}
	Let $ N = 2^{k} $ and $ \mathcal{D} $ be a $ m $-state  DFA solving $ A^{k} $ exactly.
	We show that $ m $ cannot be less than $ 2N $.
	
	Since both $ A_{yes}^{k} $ and $ A_{no}^{k} $ contain infinitely many unary strings,
	there must be a chain of $ t $ states, say $ s_{0}, \ldots, s_{t-1} $ such that,
	for sufficiently long strings, $ \mathcal{D} $ enters this chain in which
	$ \mathcal{D} $ transmits from $ s_{i} $ to $ s_{(i+1 \mod t)} $ when reading an $ a $, 
	where $ 0 \leq i \leq t-1 $ and $ 0 < t \leq m $.
	
	Without lose of generality, we assume that
	$ \mathcal{D} $ accepts the input if it is in $ s_{0} $ before reading $ \dollar $.
	Thus, $ \mathcal{D} $ rejects the input if it is in $ s_{(N \mod t)} $ before reading $ \dollar $.
	Let $ S_{a} $ be the set of $ \{ s_{( i2N \mod t)}  \mid i \geq 0 \} $.
	Then, $ \mathcal{D} $ accepts the input if it is in one of the states in $ S_{a} $ before reading $ \dollar $.
	Note that $ s_{(N \mod t)} \notin S_a $.
	
	Let $ d = \gcd(t,2N) $, $ t' = \frac{t}{d} $, and $ S' $ be the set $ \{s_{id} \mid 0 \leq i < t' \} $.
	Since $ S_{a} \subseteq S' $ and  $ |S'| = t' $,
	we can easily follow $ S_{a} = S' $ if we show $ |S_{a}| \geq t' $.

	Firstly, we show that each $ i $ satisfying $ (i2N \equiv 0 \mod t) $ must be a multiple of $ t' $:
	For such an $ i $, there exists a $ j $ such that  $ i2N = jt $.
	By dividing both sides with $ t = dt' $, we get
	$ \frac{i}{t'}\frac{2N}{d} = j $. This implies that $ i $ must be a multiple of $ t' $
	since left side must be an integer and $ gcd(t',2N) = 1 $.
	
	Secondly, we show that there is no $  i_{1} $ and $ i_{2} $, i.e.
	$ t' >  i_{1} > i_{2} \geq 0 $, such that $ (i_{1}2N \equiv i_{2}2N \mod t) $.
	If so, we have $ (i_{1}2N - i_{2}2N \equiv 0 \mod t) $ and then
	$ ((i_{1}-i_{2})2N \equiv 0  \mod t) $. This implies that $ (i_{1}-i_{2}) $
	must be a multiple of $ t' $. This is a contradiction.

	Thus, for each $ i \in \{0,\ldots,t'-1\} $, we obtain a different value of $ (i2N \mod t) $
	and so $ |S_{a}| $ contains at least $ t' $ elements.
	
	If $ \gcd(t,N) = d $, then $ s_{(N \mod t)} $ also becomes a member of $ S_{a} $.
	Therefore, $ \gcd(t,N) $ must be different than $ \gcd(t,2N) $.
	This can only be possible whenever $ t $ is a multiple of $ 2N $.
	Therefore, $ m $ cannot be less than $ 2N $.
\end{proof}

Since a $ 2^{k+1} $-state DFA solving promise problem $ A^{k} $ exactly can be constructed in a straightforward way,
we obtain the following theorem.

\begin{theorem}
	The minimal DFA solving the promise problem $ A^{k}=( A_{yes}^{k},A_{no}^{k} ) $ exactly has $ 2^{k+1} $ states.
\end{theorem}

% SSSSSSSSSSSSSSSSSSSSSSSSSSSSSSSSSSSSSSSSSSSSSSSSSSSSSSSSSSSSSSSSSSSSSSSSSSSSSSSS %
% SSSSSSSSSSSSSSSSSSSSSSSSSSSSSSSSSSSSSSSSSSSSSSSSSSSSSSSSSSSSSSSSSSSSSSSSSSSSSSSS %
% SSSSSSSSSSSSSSSSSSSSSSSSSSSSSSSSSSSSSSSSSSSSSSSSSSSSSSSSSSSSSSSSSSSSSSSSSSSSSSSS %
\section{Concluding remarks}
% SSSSSSSSSSSSSSSSSSSSSSSSSSSSSSSSSSSSSSSSSSSSSSSSSSSSSSSSSSSSSSSSSSSSSSSSSSSSSSSS %
% SSSSSSSSSSSSSSSSSSSSSSSSSSSSSSSSSSSSSSSSSSSSSSSSSSSSSSSSSSSSSSSSSSSSSSSSSSSSSSSS %
% SSSSSSSSSSSSSSSSSSSSSSSSSSSSSSSSSSSSSSSSSSSSSSSSSSSSSSSSSSSSSSSSSSSSSSSSSSSSSSSS %

In this paper, we identify a case in which the superiority of quantum computation to classical one cannot be bounded.
For this purpose, we use an infinite family of two unary disjoint languages containing the strings of the form 
$ (a^{2n})^{*} $ and $ a^{n}(a^{2n})^{*} $, respectively, where $ n $ is a power of $ 2 $.

What happens if $ n $ is not an exact power of 2?
For quantum case, we can still solve the same problem with 2 states.
On the other hand, for the classical case, the minimum number of states is determined by the biggest factor of the number,
which is a power of 2.
Let $ k,l > 0 $.
Let $N = 2^{k}(2l+1) $
and $ A^{N} = (A_{yes}^{N},A_{no}^{N}) $ 
(where $ A_{yes}^{N} = \{ a^{iN} \mid i \mbox{ is a nonnegative \textit{even} integer} \} $ and
$ A_{no}^{N} = \{ a^{iN} \mid i \mbox{ is a positive \textit{odd} integer} \} $) be a promise problem.

\begin{corollary}
	The minimal DFA solving promise problem $ A^{N} = (A_{yes}^{N},A_{no}^{N}) $ exactly has $ 2^{k+1} $ states.
	\footnote{The proof can be obtained by using almost the same technique given in Section 2.}	
\end{corollary}
\noindent
Therefore, if $ N $ is an odd integer, a DFA only needs 2 states to solve the related promise problems.

\bibliographystyle{alpha}
\bibliography{YakaryilmazSay}

\newcommand{\etalchar}[1]{$^{#1}$}
\begin{thebibliography}{BCdWZ99}

\bibitem[AF98]{AF98}
Andris Ambainis and R\={u}si\c{n}\v{s} Freivalds.
\newblock 1-way quantum finite automata: strengths, weaknesses and
  generalizations.
\newblock In {\em FOCS'98: Proceedings of the 39th Annual Symposium on
  Foundations of Computer Science}, pages 332--341, 1998.

\bibitem[BBC{\etalchar{+}}98]{BHCMW98}
Robert Beals, Harry Buhrman, Richard Cleve, Michele Mosca, and Ronald de~Wolf.
\newblock Quantum lower bounds by polynomials.
\newblock In {\em FOCS'98: Proceedings of the 39th Annual Symposium on
  Foundations of Computer Science}, pages 352--361, 1998.

\bibitem[BCdWZ99]{BCWZ99}
Harry Buhrman, Richard Cleve, Ronald de~Wolf, and Christof Zalka.
\newblock Bounds for small-error and zero-error quantum algorithms.
\newblock In {\em FOCS'99: Proceedings of the 40th Annual Symposium on
  Foundations of Computer Science}, pages 358--359, 1999.

\bibitem[BdW03]{BV03}
Harry Buhrman and Ronald de~Wolf.
\newblock Quantum zero-error algorithms cannot be composed.
\newblock {\em Information Processing Letters}, 87:79--84, 2003.

\bibitem[BH97]{BH97}
Gilles Brassard and Peter Hoyer.
\newblock An exact quantum polynomial-time algorithm for simon's problem.
\newblock In {\em ISTCS'97: Proceedings of the Fifth Israel Symposium on the
  Theory of Computing Systems}, pages 12--23, 1997.

\bibitem[BMP03]{BMP03}
Alberto Bertoni, Carlo Mereghetti, and Beatrice Palano.
\newblock Quantum computing: 1-way quantum automata.
\newblock In Zolt\'{a}n \'{E}sik and Zolt\'{a}n F\"{u}l\"{o}p, editors, {\em
  Developments in Language Theory}, volume 2710 of {\em LNCS}, pages 1--20,
  2003.

\bibitem[BV97]{BV97}
Ethan Bernstein and Umesh Vazirani.
\newblock Quantum complexity theory.
\newblock {\em SIAM Journal on Computing}, 26(5):1411--1473, 1997.

\bibitem[Cia01]{Ci01}
Massimo~Pica Ciamarra.
\newblock Quantum reversibility and a new model of quantum automaton.
\newblock In {\em FCT'01: Proceedings of the 13th International Symposium on
  Fundamentals of Computation Theory}, pages 376--379, 2001.

\bibitem[FI09]{FI09}
R\={u}si\c{n}\v{s} Freivalds and Kazuo Iwama.
\newblock Quantum queries on permutations with a promise.
\newblock In {\em CIAA'09: Proceedings of the 14th International Conference on
  Implementation and Application of Automata}, pages 208--216, 2009.

\bibitem[Hir10]{Hi10}
Mika Hirvensalo.
\newblock Quantum automata with open time evolution.
\newblock {\em International Journal of Natural Computing Research},
  1(1):70--85, 2010.

\bibitem[Kla00]{Kl00}
Hartmut Klauck.
\newblock On quantum and probabilistic communication: Las vegas and one-way
  protocols.
\newblock In {\em STOC'00: Proceedings of the thirty-second annual ACM
  symposium on Theory of computing}, pages 644--651, 2000.

\bibitem[KW97]{KW97}
Attila Kondacs and John Watrous.
\newblock On the power of quantum finite state automata.
\newblock In {\em FOCS'97: Proceedings of the 38th Annual Symposium on
  Foundations of Computer Science}, pages 66--75, 1997.

\bibitem[MC00]{MC00}
Cristopher Moore and James~P. Crutchfield.
\newblock Quantum automata and quantum grammars.
\newblock {\em Theoretical Computer Science}, 237(1-2):275--306, 2000.

\bibitem[MNYW05]{MNYW05}
Yumiko Murakami, Masaki Nakanishi, Shigeru Yamashita, and Katsumasa Watanabe.
\newblock Quantum versus classical pushdown automata in exact computation.
\newblock {\em IPSJ Digital Courier}, 1:426--435, 2005.

\bibitem[Wat09]{Wa09}
John Watrous.
\newblock Quantum computational complexity.
\newblock In Robert~A. Meyers, editor, {\em Encyclopedia of Complexity and
  Systems Science}, pages 7174--7201. Springer, 2009.

\bibitem[YFSA10]{YFSA10}
Abuzer Yakary{\i}lmaz, R\={u}si\c{n}\v{s} Freivalds, A.~C.~Cem Say, and Ruben
  Agadzanyan.
\newblock Quantum computation with devices whose contents are never read.
\newblock In {\em Unconventional Computation}, volume 6079 of {\em Lecture
  Notes in Computer Science}, pages 164--174, 2010.

\bibitem[YS11]{YS11A}
Abuzer Yakary{\i}lmaz and A.~C.~Cem Say.
\newblock Unbounded-error quantum computation with small space bounds.
\newblock {\em Information and Computation}, 209(6):873--892, 2011.

\end{thebibliography}

\end{document}